\documentclass[letterpaper]{article}
\usepackage{times} \usepackage{helvet} \usepackage{courier}
\usepackage{graphicx} \usepackage{subfig} \usepackage{amsfonts}
\usepackage{amsmath} \usepackage{multirow} \usepackage{amssymb}
\usepackage{amsthm}  
\newtheorem{theorem}{Theorem} \newtheorem{example}{Example}
\newtheorem{lemma}{Lemma} \newtheorem{corollary}{Corollary}
\newtheorem{definition}{Definition}

\begin{document} 
\title{On Covering Codes and Upper Bounds for the Dimension of Simple Games}
\author{Martin Olsen\\ BTECH, Aarhus University\\ Denmark\\ martino@btech.au.dk\\ }

\maketitle

\begin{abstract} 
Consider a situation with $n$ agents or players where some of the
players form a coalition with a certain collective objective. Simple
games are used to model systems that can decide whether coalitions
are successful (winning) or not (losing). A simple game can be
viewed as a monotone boolean function. The dimension of a simple
game is the smallest positive integer $d$ such that the simple game
can be expressed as the intersection of $d$ threshold functions
where each threshold function uses a threshold and $n$ weights.
Taylor and Zwicker have shown that $d$ is bounded from above by the
number of maximal losing coalitions. We present two new upper bounds
both containing the Taylor/Zwicker-bound as a special case. The
Taylor/Zwicker-bound imply an upper bound of ${n \choose
\frac{n}{2}}$. We improve this upper bound significantly by showing
constructively that $d$ is bounded from above by the cardinality of
any binary covering code with length $n$ and covering radius $1$.
This result supplements a recent result where Olsen et al. showed
how to construct simple games with dimension $|C|$ for any binary
constant weight SECDED code $C$ with length $n$. Our
result represents a major step in the attempt to close the
dimensionality gap for simple games.
\end{abstract}

\section{Introduction}

Consider a multi agent system where a coalition of agents is formed
in order to solve a given task and where we have to predict if the
coalition will succeed or not. We restrict our attention to cases
obeying the natural monotonicity condition saying that the superset
of any successful coalition will also succeed. In such a multi agent
system we need some sort of system that can compute a prediction:
''yes'' or ''no''. The so-called {\em simple games} model such
systems and simple games can also be viewed as monotone boolean
functions or monotone hypergraphs. The agents in a simple game are
referred to as players and successful and unsuccessful coalitions
will be referred to as winning and losing coalitions respectively.

A {\em weighted game} is a special type of simple game where every
player is assigned a weight and where a coalition is successful if
and only if the total weight of the players in the coalition is
meeting or exceeding a given quota. Any simple game can be
implemented as the intersection of one or more weighted games and
the {\em dimension}~\cite{TaZw99} of a simple game is the minimum
number of weighted games we need to implement the simple game in
this way. The dimension has a direct influence on the storage
requirements and the efficiency for such a system. Real world voting
systems can be seen as simple games and the dimension aspects of
real world voting systems have been studied intensively within the
field of {\em Computational Social Choice}.

In this paper, we consider the maximum dimension, $d_n$, that we can
obtain for a simple game with $n$ players.
Taylor and Zwicker~\cite{TaZw99} have shown that ${n\choose
{\lfloor n/2\rfloor}}$ is an upper bound for $d_n$ by demonstrating
how to implement any simple game as the intersection of no more than
${n\choose {\lfloor n/2\rfloor}}$ games (details will follow later).
The main contribution of this paper is a constructive major
improvement of the generic upper bound provided by
Taylor and Zwicker that we present in the form of two new upper
bounds both representing stronger versions of the upper bound
presented by Taylor and Zwicker.

We apply a technique that -- to the best of our knowledge -- has not
been used before to translate any simple game into the intersection
of relatively few simple games. Recently,
Olsen et al.~\cite{Olsen2016} demonstrated a major
improvement in the lower bound on $d_n$ by using theory on {\em
error correcting codes}. We use a significantly different and novel
approach based on {\em covering codes} to obtain our upper bounds.
The gap between the upper and lower bound has for some $n$ gone from
a factor $n$ to $\sqrt{n}$ (roughly) through our improvement and to
a factor $\ln n \sqrt{n}$ in general. We conclude by suggesting a
direction that might lead to a further reduction of the
dimensionality gap.
  
\section{Related Work}

Taylor and Zwicker~\cite{TaZw99} have constructed a sequence
of games with dimension at least $2^{\frac{n}{2}-1}$ for $n=2k$ with
$k$ odd. The dimension of the simple games presented by
Taylor and Zwicker was later shown to be exactly
$2^{\frac{n}{2}-1}$~\cite{Olsen2016}.
Freixas and Puente~\cite{Freixas2001265} have shown
how to construct another type of simple games with dimension
$2^{\frac{n}{2}-1}$ for all even $n$. This lower bound of $d_n$ was
recently improved significantly by
Olsen et al.~\cite{Olsen2016} by establishing a
connection to the theory on error correcting codes resulting in the
following lower bound:
\begin{equation}\label{eq:lowerbound}
d_n \geq \frac{1}{n} {n \choose \lfloor \frac{n}{2} \rfloor} \in 2^{n-o(n)}\enspace .
\end{equation}
Here it might be useful to consider the following identity for
comparison with the previous lower bound:
\begin{equation}\label{eq:nchoosen2}
{n \choose \lfloor \frac{n}{2} \rfloor} = (1-o(1))\sqrt{\frac{2}{\pi n}} 2^{n} \enspace .
\end{equation}
Kurz and Napel~\cite{Kurz2015} also present a general approach for the determination of lower bounds for
the dimension of a simple game.

A maximal losing coalition in a simple game is a losing coalition
that has the property that adding any player will turn it into a
winning coalition. Let $L^M$ denote the collection of maximal
losing coalitions. Taylor and Zwicker~\cite{TaZw99}
demonstrate how to express any simple game as an intersection of at
most $|L^M|$ weighted games implying an $|L^M|$-upper bound for
$d_n$. Kurz and Napel~\cite{Kurz2015} provide heuristic
algorithms based on integer linear programming for constructing a
representation of a given simple game as an intersection of weighted
games.

As mentioned earlier, the dimension of real world voting systems has
been the focus for several studies. The Amendment of the Canadian
constitution \cite{Kilgour1983} and the US federal legislative
system \cite{Taylor1993170} have dimension $2$. The voting systems
of the Legislative Council of Hong Kong~\cite{Cheung2014} and the
Council of the European Union under its Treaty of Nice
rules~\cite{Freixas2004415} have a dimension of exactly $3$.
Kurz and Napel~\cite{Kurz2015} have established that the
dimension of the voting system of the Council of the European Union
under its Treaty of Lisbon rules is between $7$ and $13\,368$.

There are obviously alternative ways for representing simple games.
The {\em codimension} ~\cite{NotionsOfDimension} is the minimum
number of weighted games it takes to represent a simple game as a
union of weighted games. Considering arbitrary combinations of
unions and intersections leads to the notion of {\em boolean
dimension}, which is introduced and studied in
\cite{BooleanCombinations}.

\subsection{Outline of the Paper}

The next section introduces the notation and the formal definitions
for simple games. We also give a brief introduction for readers not
familiar with covering codes. The algorithm behind our first upper
bound on $d_n$ is then presented in two sections. The first of the
sections demonstrates how the algorithm works and the second section
contains the technical details and proofs including a formal
statement of the upper bound in terms of a theorem. The second upper
bound and our second theorem is then presented in a section and
finally, we wrap the paper up in the conclusion.

\section{Preliminaries} 

In this section, we introduce the concepts and definitions that we
consider in this paper. We start by presenting formal definitions
for simple games. After that we give a brief introduction to
covering codes.

\subsection{Simple Games}
\label{subsec_simple_games}

We now formally the define simple games:
\begin{definition}
A simple game $\Gamma=(N, W)$ is a pair where $N = \{1,\dots,n\}$
for some positive integer $n$
and $W\subseteq 2^N$ is a collection of subsets of $N$ such that:
\begin{itemize} 
\item $\emptyset\notin W$
\item $N\in W$
\item $S\subseteq T\subseteq N$ and $S\in W$ implies $T\in W$
\end{itemize}
\end{definition} 
The members of $N$ are referred to as players and subsets of $N$ are
referred to as coalitions. A coalition is said to be winning if it
is a member of $W$ and otherwise it is said to be losing.
The first condition says that the coalition with no players loses
and the second condition ensures that the coalition containing all
players wins. The third condition is the monotonicity condition that
says that any superset of a winning coalition is also winning. The
set of losing coalitions is denoted by $L = 2^N \setminus W$.

A coalition is a maximal losing coalition if it is losing and all of
its supersets are winning. The collection of coalitions $L^M \subset
2^N$ contains all the maximal losing coalitions. The collection of
minimal winning coalitions $W^m$ is defined accordingly. A simple
game $\Gamma$ can be defined by either of the sets $W$, $L$, $W^m$
or $L^M$ .

The weighted games that form a proper subset of the simple games are defined as follows:
\begin{definition}
A simple game $\Gamma=(N,W)$ is weighted if there exists a
\emph{quota} $q \in \mathbb{R}_+$ and \emph{weights} $w_1, w_2,
\ldots , w_n \in \mathbb{R}_+$ such that $S \in W$ if and only
if $\sum_{i \in S} w_i \geq q$. In this case we use the notation $\Gamma=[q;w_1,w_2,\ldots,w_n]$.
\end{definition}

The intersection $\Gamma_1 \cap \Gamma_2$ of the games $\Gamma_1(N,
W_1)$ and $\Gamma_2(N, W_2)$ is the simple game with players $N$ and
$W = W_1 \cap W_2$. As previously mentioned,
Taylor and Zwicker~\cite{TaZw99} have shown that any simple
game can be expressed as the intersection of $|L^M|$ weighted games:
For any game $\Gamma$, we have $\Gamma = \cap_{T \in L^M} \Gamma_T$
where a coalition $S$ wins in $\Gamma_T$ if and only if $S \not
\subseteq T$. A weighted representation of $\Gamma_T$ using weights
$0$ and $1$ is given as follows: the game has quota $1$ and a player
in $N \setminus T$ is assigned the weight $1$ and all other players
are assigned weight $0$.

The dimension of a simple game can now be formally defined:
\begin{definition}
The dimension $d$ of a simple game $\Gamma$ is the smallest positive
integer such that $\Gamma = \cap_{i = 1}^d \Gamma_d$ where the games
$\Gamma_i$, $i \in \{ 1, 2, \ldots, d\}$, are weighted.
\end{definition}
In this paper, we let $d_n$ denote the maximum dimension that we can
observe for a simple game with $n$ players.

A maximal losing coalition cannot contain another maximal losing
coalition, so we can apply Sperner's Lemma~\cite{Lubell1966} and get
the following upper bound on $L^M$: $\left|L^M\right|\le {n \choose
{\lfloor n/2\rfloor}}$. From the construction by Taylor and Zwicker, we conclude the following:
\begin{equation}
\label{eq:d-upperbound} 
d_n \leq |L^M| \leq {n \choose \lfloor \frac{n}{2} \rfloor}
\enspace .
\end{equation}
The main objective of this paper is to improve this upper bound.

We will illustrate the definitions by an example.
\begin{example}
Let the simple game $\Gamma(N,L^M)$ be defined as follows:
$$N = \{1, 2, 3, 4, 5, 6, 7\}$$
$$L^M = \{\{1, 2, 3\}, \{3, 4, 5, 6\}\} \enspace .$$
The coalition $\{1, 2\}$ loses in $\Gamma$ since $\{1, 2\} \subseteq \{1, 2, 3\}$. The coalition $\{1, 4\}$ wins since $\{1, 4\} \not \subseteq \{1, 2, 3\}$ and $\{1, 4\} \not \subseteq \{3, 4, 5, 6\}$. 

If we use the construction by Taylor and Zwicker, we get this representation of $\Gamma$ as the intersection of two weighted games:
$$\Gamma = [1;0,0,0,1,1,1,1] \cap [1;1,1,0,0,0,0,1] \enspace .$$
The dimension of $\Gamma$ is $2$ since $\Gamma$ cannot be weighted. We can realize this using a proof by contradiction that illustrates a classical way of establishing lower bounds for the dimension:

Assume that $\Gamma$ was weighted with quota $q$. The coalitions $\{1,2\}$ and $\{4,5\}$ are both losing so the total weight of the players in the coalitions must be strictly smaller than $2q$. The two coalitions can exchange players and both win after the exchange: $\{1,4\}$ and $\{2,5\}$. We now arrive at a contradiction since the total weight of the players must be at least $2q$.
\end{example}

We now turn our attention to covering codes.

\subsection{Covering Codes}

A binary {\em code} is technically a set of bit vectors. A bit
vector $x = x_1x_2 \ldots x_n \in \{0,1\}^n$ can be viewed as the coalition $S_x = \{i
\in N: x_i = 1\}$. We will use this perspective and see a binary
code as a collection of coalitions in order to align the notation of
binary codes and simple games. The Hamming distance between two bit vectors is the number of coordinates where the two bit vectors differ. Using the perspective just described, we can define the Hamming distance between two coalitions $x$ and $y$ as follows:
$$d(x,y) = |x \setminus y|+|y \setminus x| \enspace .$$

A binary {\em covering
code}~\cite{Cohen1997} of length $n$ and {\em covering radius} $1$ can
consequently be perceived as a collection $C \subset 2^N$ of coalitions
such that any coalition is within Hamming distance $0$ or $1$ from
at least one member of $C$: $\forall x \in 2^N \exists c \in C:
d(x,c) \leq 1$. As an example, covering codes have applications within data compression. In this paper, $K_n$ denotes the
minimum cardinality of a binary covering code of length $n$ with
covering radius $1$ 

\begin{example}
The following set represents a binary covering code with length $4$
and covering radius $1$:
$$C = \{ \{\}, \{4\}, \{1, 2, 3\},\{1,2,3,4\}\}$$
As an example, the coalition $\{2, 4\}$ is covered by the coalition $\{4\}$ in $C$ since the Hamming distance between these coalitions is $1$.

It is not possible to cover all subsets of
$\{1,2,3,4\}$ with fewer coalitions -- we cannot cover more than $3
\cdot 5 = 15$ coalitions with $3$ coalitions and there are $16$
coalitions in total so our example shows that $K_4 = 4$.
\end{example}

As we saw earlier, a coalition cannot cover more than $n+1$ other coalitions including
itself within radius $1$ so we need at least $2^n/(n+1)$ coalitions
for a binary covering code with covering radius $1$. The well known
Hamming codes~\cite{berlekamp2015algebraic} defined for $n=2^m-1$
are so called {\em perfect} codes that meet this lower bound. For
$n=2^m$, we have the slightly smaller value in the denominator: $K_n
= 2^n/n$~\cite{OSTERGARD1998}. In general, it is hard to
establish exact values for $K_n$ but it is not hard to prove the
upper bound $K_n \leq (\ln (n+1) + 1) 2^n/(n+1)$ using a classical
result\footnote{Consider the graph where we have a vertex for each
coalition and an edge between two vertices if and only if the
distance between the corresponding coalitions is $1$. A covering
code corresponds to a dominating set in this graph where all
vertices have degree $n$. Alon and Spencer present a
lower bound for the size of such a dominating set. This upper bound is
probably well known within the coding theory community.} from
Alon and Spencer~\cite{alon:probabilistic} on
computing dominating sets.

\section{The First Upper Bound}

From now on, any simple game will be defined using maximal losing coalitions. Given a simple game $\Gamma(N, L^M)$, we now present an algorithm
producing a representation of $\Gamma$ as an
intersection of no more than $K_n$ weighted games. In this section, we
will show how the algorithm works step by step. Each step will
contain a formal explanation but we will also illustrate how each
step works through an example. The technical details including the proof of
correctness will follow in the next section.

The key idea for the algorithm is the result of a simple observation expressed by the following lemma:
\begin{lemma}
\label{lem:keyidea}
If $L^M = \cup_{i = 1}^p L_i$ then $$\Gamma(N,L^M) = \cap_{i = 1}^p \Gamma(N,L_i) \enspace .$$
\end{lemma}
\begin{proof}
Now assume that $x \subseteq N$ is losing in $\Gamma(N,L^M)$. There
must be an $y \in L^M$ such that $x \subseteq y$. This means that
$x$ loses in any of the games $\Gamma(N,L_i)$ with $y \in L_i$. On the other
hand, $x$ will lose in $\Gamma(N,L^M)$ if $x$ loses in $\cap_{i =
1}^p \Gamma(N,L_i)$ since $x$ must be a subset of at least one $y$
that is a member of $L^M$.
\end{proof}
The objective for our algorithm is to use the lemma and partition $L^M$ into a
small number of sets such that all the corresponding games are weighted.

The game that we use as an example is the following simple game with
players $\{1, 2, 3, 4\}$: A coalition wins if and only if both of
the players $1$ and $2$ are members of the coalition or both of the
players $3$ or $4$ are members. As an example, the coalition
$\{1,2,4\}$ wins since both of the players $1$ and $2$ have joined
the coalition. On the other hand, the coalition $\{1,3\}$ is losing
-- and in fact it is a maximal losing coalition since this coalition
will turn into a winning coalition if any of the other players joint
it. As a side remark, this game belongs to a class of simple games that has been
studied in detail by Freixas and Puente~\cite{Freixas2001265}.

We are now ready to describe how our algorithm works:

\subsection{Input} The input to the algorithm is a simple game $\Gamma(N,L^M)$. Example:
$$N = \{1, 2, 3, 4\}$$ $$L^M = \{\{1, 3\}, \{1, 4\}, \{2, 3\}, \{2,
4\}\}$$

\subsection{Step $1$} Construct a collection of coalitions $C
\subset 2^N$ such that any coalition in $L^M$ is within Hamming distance $0$
or $1$ from at least one coalition in $C$: $\forall x \in L^M
\exists c \in C: d(x,c) \leq 1$. Example: $$C = \{ \{4\}, \{1, 2, 3\}\}$$

\subsection{Step $2$} Let $\{L_c\}_{c \in C}$ be a partition
of $L^M$ such that all members of
$L_c$ have distance $0$ or $1$ to $c$: $\forall x \in L_c: d(x,c)
\leq 1$. Example:  $$L_{\{4\}}=\{\{1,
4\}, \{2, 4\}\}$$ $$L_{\{1, 2, 3\}}=\{\{1, 3\}, \{2, 3\}\}$$

\subsection{Step $3$} For each $c \in C$, we now represent
$\Gamma(N,L_c)$ as a weighted game $[q^c; w^c_1, w^c_2, \ldots,
w^c_n]$. We prove that $\Gamma(N,L_c)$ is weighted for any $c \in C$ and provide the details on how to compute the weights and the quota
in Lemma~\ref{lem:weighted} below. Example: $$\Gamma(N,
L_{\{4\}})=[2;1,1,2,0]$$ $$\Gamma(N,L_{\{1, 2, 3\}})=[2;1,1,0,2]$$

\subsection{Output} Finally, we can use Lemma $\ref{lem:keyidea}$ and express $\Gamma$ as the
intersection of the weighted games that we have constructed in Step
$3$: $ \Gamma = \cap_{c \in C} [q^c; w^c_1, w^c_2, \ldots,
w^c_n]$. Example: $$\Gamma = [2;1,1,2,0] \cap [2;1,1,0,2]$$
This concludes the description of our algorithm.

In Step $1$, we can actually use a binary covering code of length $n$ with covering radius $1$ for {\em any} simple game involving $n$ players. As a consequence, any simple game can be implemented as the intersection of no more than $K_n$ weighted games. This allows us to set up the following upper bounds on $d_n$ using the facts on $K_n$ from the previous section:
\begin{equation}
\label{eq:kn1}
d_n \leq K_n = \frac{2^n}{n+1} \mbox{ for } n = 2^m-1
\end{equation}
\begin{equation}
\label{eq:kn2}
d_n \leq K_n = \frac{2^n}{n} \mbox{ for } n = 2^m
\end{equation}
\begin{equation}
\label{eq:kn3}
d_n \leq K_n \leq (\ln(n+1)+1) \frac{2^n}{n+1} \mbox{ for all } n \enspace .
\end{equation}
In all three cases, the upper bounds are considerably smaller than
${n \choose \lfloor \frac{n}{2} \rfloor}$ which can be seen from
(\ref{eq:nchoosen2}). The first two upper bounds represent an
improvement on roughly a factor $\sqrt{n}$ and the bounds are all
$o({n \choose \lfloor \frac{n}{2} \rfloor})$.

It is important to observe that it might be a bad idea to use a
binary covering code as a ''one size fits all''-solution since we do
not exploit the structure of $L^M$ for the specific game at hand if
we follow this approach.

Osterg{\aa}rd and Kaikkonen~\cite{OSTERGARD1998} have listed some
upper bounds for $K_n$ that we also can use as upper bounds for
$d_n$. Table~\ref{tab:bounds} presents these upper bounds together
with lower bounds from~\cite{Olsen2016}.

\begin{table}
\centering
\caption{This table displays lower and upper bounds for $d_n$ combining our findings with the results from \cite{OSTERGARD1998} and \cite{Olsen2016}.}
\begin{tabular}{|r|c|c|c|}
\hline
$n$ & Lower bound & Upper bound & ${{n}\choose{{\lfloor n/2\rfloor}}}-1$ \\
\hline
6 &  4 &  12 & 19\\
7 &  7 &  16 & 34\\
8 &  14 &  32 & 69\\
9 &  18 &  62 & 125\\
10 & 36 &  120 & 251\\
11 & 66 &  192 & 461 \\
12 & 132 &  380 & 923\\
13 & 166 &  704 & 1715\\
14 & 325 &  1408 & 3431\\
15 & 585 &  2048  & 6434\\
\hline
\end{tabular}
\label{tab:bounds}
\end{table}

\section{Technical Details for the First Upper Bound}
\label{sec:technicaldetails}

We now take another look at our approach where we formally prove our
first upper bound and state the bound as a theorem.

We have to ensure is that our algorithm is
correct in the sense that it is able to express any input game as an
intersection of weighted games. It is clearly possible to produce
the collection $C$ in Step $1$ and to construct the partition of
$L^M$ in Step $2$. The algorithm uses the decomposition approach
suggested by Lemma $\ref{lem:keyidea}$ so we only have to check that
all the games considered in Step $3$ are weighted.

\begin{lemma}
\label{lem:weighted}
$\Gamma(N,L_c)$ is weighted for any $c \in C$.
\end{lemma}

\begin{proof}
All the members of $L_c$ are maximal losing coalitions so it is not possible to find two members of $L_c$ such that one of them contains the other. This means there are three cases that we have to consider: 1) $\forall x \in L_c: x \subset c$, 2) $L_c = \{c\}$, or 3) $\forall x \in L_c: c \subset x$. We now show how to express $\Gamma(N,L_c)$ as a weighted game in all three cases.

Case $1$: The set $L_c$ consists of coalitions were exactly one
element has been removed from $c$ for each member of $L_c$. Let $R$
denote the set of removed elements: $R = \cup_{x \in L_c} (c
\setminus x)$. Let us consider a set $S$ that is winning and is
contained in $c$. For any $x \in L_c$, we know that $S$ is not
contained in $x$ so $S$ must contain the element that has been
removed from $c$ to form $x$. In other words, $S$ cannot win in
$\Gamma(N,L_c)$ unless $S \setminus c \neq \emptyset$ or $R
\subseteq S$. On the other hand, it is not hard see that $S$ wins if
$S \setminus c \neq \emptyset$ or $R \subseteq S$. This means that
we can implement $\Gamma(N,L_c)$ as the weighted game with $q = |R|$
and weights as follows: $w_i = |R|$ for $i \not \in c$, $w_i = 1$
for $i \in R$ and $w_i =0$ for the remaining players.

As an example, we consider the game
$\Gamma(N,L_c)$ with $N=\{1, 2, 3, 4, 5\}$, $c=\{1, 2, 3, 4\}$ and
$L_c = \{\{1, 2, 3\}, \{1, 2, 4\}, \{1, 3, 4\}\}$. For this game we
have $R = \{2,3,4\}$ and $\Gamma(N,L_c) = [3;0,1,1,1,3]$.

Case $2$: In this case, we can use the weighted game with quota $q=1$
where we assign the weight $0$ to all players in $c$ and the weight
$1$ to all other players.

Case 3: All the members of $L_c$ are constructed by adding exactly
one element to $c$. Let $A$ denote the set of added elements: $A =
\cup_{x \in L_c} (x \setminus c)$. If a coalition $S$ wins and $S$
does not contain any players in $c \cup A$ then $S$ has to contain
at least two players in $A$ (otherwise $S$ would lose). Conversely,
$S$ wins if $S$ contains a player not in $c \cup A$ or at least two
of the players in $A$. This implies that $\Gamma(N,L_c)$ can be
expressed as a weighted game with quota $q=2$ and the following
weight distribution: $w_i = 2$ for $i \not \in c \cup A$, $w_i = 1$
for $i \in A$ and $w_i =0$ for the players in $c$.

An example for case $3$: $\Gamma(N,L_c)$ with $N=\{1, 2, 3, 4, 5, 6, 7\}$, $c=\{1, 2, 3\}$ and
$L_c = \{\{1, 2, 3, 4\}, \{1, 2, 3, 5\}, \{1, 2, 3, 6\}\}$. Here we
have $A = \{4,5,6\}$ and $\Gamma(N,L_c) = [2;0,0,0,1,1,1,2]$.
\end{proof}

We are now ready to formally state the main contribution of our paper:
\begin{theorem}
\label{thm:bound1}
Let $\Gamma(N,L^M)$ be a simple game and let $C \subset 2^N$ be a collection of coalitions such that $\forall x \in L^M \exists c \in C: d(x,c) \leq 1$. The dimension of $\Gamma(N,L^M)$ is bounded from above by $|C|$.  
\end{theorem}
\begin{proof}
We can use our algorithm to produce a representation of $\Gamma$ as the intersection of $|C|$ weighted games. Lemma $\ref{lem:keyidea}$ and Lemma $\ref{lem:weighted}$ guarantee that our algorithm is correct.
\end{proof}
It is important to note that the special case $C = L^M$ corresponds to the $|L^M|$-upper bound presented by Taylor and Zwicker~\cite{TaZw99}.

If we have a binary covering code with covering radius $1$ then we
can use it as $C$ in the theorem. We therefore have the following
corollary:
\begin{corollary}
$$d_n \leq K_n$$
\end{corollary}
It is important to stress that we only require $C$ to ''cover'' the
set $L^M$ in the theorem above. We might be able to exploit the
structure of $L^M$ in order to achieve a better upper bound than in
the corollary where the underlying collection covers all possible
coalitions. As an example, we might use the fact that $L^M$ is a
Sperner family where no member contains another member of the
family. This explains why we have chosen to express the bound $d_n
\leq K_n$ as a corollary since the theorem is a stronger result.

\section{The Second Upper Bound}

In this section, we will once again use the key idea from
Lemma~\ref{lem:keyidea} and prove another upper bound generalizing
the $|L^M|$-upper bound presented by
Taylor and Zwicker~\cite{TaZw99}. This upper bound is related
to SECDED codes that are binary codes where any two of the members
have pairwise distance at least $4$.

\begin{theorem}
\label{thm:bound2}
Let $\Gamma(N,L^M)$ be a simple game. The dimension of
$\Gamma(N,L^M)$ is bounded from above by $\frac{1}{2}(|L^M|+|C|)$
for some collection $C \subseteq L^M$ of maximal losing coalitions
satisfying $\forall x, y \in C: d(x,y) \geq 4$.
\end{theorem}
\begin{proof}
Let $M$ be a maximal set of pairs $(x,y) \in L^M \times L^M$ such
that $x \neq y$ and $d(x,y) \leq 3$ and such that an element in
$L^M$ occurs in no more than one pair. We claim that the game
$\Gamma(N,\{x,y\})$ is weighted for any $(x,y) \in M$. We will
prove it for the case $d(x,y) = 3$ and leave the only remaining case
$d(x,y) = 2$ to the reader (there are no more cases since $L^M$ is a Sperner family).

Without loss of generality, we assume that $x \setminus y$ contains
two players and $y \setminus x$ contains one player. A coalition
wins in the game if and only if: 1) the coalition contains at least
one player in $N \setminus (x \cup y)$, or 2) the coalition contains
one of the players in $x \setminus y$ and the player in $y \setminus
x$. We implement the game $\Gamma(N,\{x,y\})$ as a weighted game
with quota $q=3$. The players in $N \setminus (x \cup y)$ get weight
$3$. The two players in $x \setminus y$ get weight $1$ and the
player in $y \setminus x$ gets weight $2$. All the players in $x
\cap y$ are assigned the weight $0$.

Let us illustrate the construction with the example with $N=\{1, 2,
3, 4, 5, 6, 7\}$, $x = \{1,2,3,4\}$ and $y = \{2,3,5\}$. The
corresponding weighted game is $[3;1,0,0,1,2,3,3]$.

Let $C$ be the set of coalitions that have not been paired in $M$.
All the coalitions in $C$ have pairwise distance at least $4$ since
$M$ is maximal. The pairs in $M$ and the coalitions in $C$
considered as single element sets constitute a partition of $L^M$
where all the corresponding games are weighted. This partition
consists of no more than $\frac{1}{2}(|L^M|-|C|)+|C|$ coalitions.
\end{proof}

A corollary of the theorem is as follows: 
\begin{corollary}
The dimension of $\Gamma(N,L^M)$ is less than $|L^M|$ if $L^M$ is not a SECDED
code.
\end{corollary}

\section{Conclusion}
We have presented two new upper bounds on the maximum dimension
$d_n$ for simple games with $n$ players. The bounds are related to
binary codes and they represent improvements of the $|L^M|$-upper
bound presented by
Taylor and Zwicker~\cite{TaZw99}.

The recent development~\cite{Olsen2016} for the lower bound of $d_n$ can be
illustrated as follows:
\begin{equation}
\label{eq:lb}
2^{\frac{n}{2}-1} \rightarrow \frac{1}{n} {n \choose \lfloor \frac{n}{2} \rfloor} = (1-o(1))\sqrt{\frac{2}{\pi n}} \frac{2^{n}}{n}  \leq d_n\enspace .
\end{equation}
On the other hand, one of the upper bounds in our paper represents the following improvement
with respect to the upper bound for $n=2^m-1$:
\begin{equation}
\label{eq:ub}
d_n\leq\frac{2^n}{n+1} \leftarrow(1-o(1))\sqrt{\frac{2}{\pi n}} 2^{n}={n \choose \lfloor \frac{n}{2} \rfloor}
\end{equation}
The dimensionality gap for the simple games is now considerably
smaller and the upper bound is roughly within a factor $\sqrt{n}$
away from the lower bound for some values of $n$.
 
As previously mentioned, we only have to cover $L^M$ with a binary
covering code with radius $1$ to obtain an upper bound on the
dimension as expressed by Theorem~\ref{thm:bound1}. It is not known
-- at least to the authors of this paper -- whether it is possible
but it seems plausible to improve the upper bound from (\ref{eq:ub})
by using the fact that $L^M$ has a certain structure.

The key idea behind our upper bounds is to decompose $L^M$ into a
union of collections of maximal losing coalitions such that any of
the simple games defined by the component collections are weighted.
This can be done in many ways and it is highly likely that there are
smarter decompositions than the ones presented in our paper. It is
an open problem to find smarter decompositions.


\bibliographystyle{plain}
\bibliography{Olsen}

\end{document}